\documentclass[12pt]{amsart}
\usepackage[T1]{fontenc}
\usepackage{amsthm}
\usepackage[margin=1.1in]{geometry}
\usepackage{enumitem}
\usepackage{amsmath}
\usepackage{cite}
\usepackage{amssymb}
\usepackage{xcolor}
\usepackage{mathrsfs}
\usepackage[
	colorlinks,
	linkcolor={black},
	citecolor={black},
	urlcolor={black}
]{hyperref}
\usepackage{todonotes}

\newcommand{\bbC}{{\mathbb{C}}}
\newcommand{\bbD}{{\mathbb{D}}}
\newcommand{\bbN}{{\mathbb{N}}}
\newcommand{\bbR}{{\mathbb{R}}}
\newcommand{\bbT}{{\mathbb{T}}}
\newcommand{\bbZ}{{\mathbb{Z}}}

\newcommand{\rmd}{{\mathrm{d}}}

\DeclareMathOperator{\sgn}{sgn}
\DeclareMathOperator{\spectrum}{spec}
\DeclareMathOperator{\separate}{sep}

\DeclareMathOperator{\diag}{diag}
\DeclareMathOperator{\imaginary}{Im}

\renewcommand{\Im}{\imaginary}
\DeclareMathOperator{\annulus}{{\mathcal{A}}}

\newcommand{\floquet}{{\mathscr{F}}}
\newcommand{\iop}{{\mathbf{i}}}
\newcommand{\LR}{{\mathrm{LR}}}
\newcommand{\asy}{{\mathrm{asy}}}

\allowdisplaybreaks
\newtheorem{theorem}{Theorem}[section]

\newtheorem{lemma}[theorem]{Lemma}
\newtheorem{coro}[theorem]{Corollary}

\theoremstyle{definition}
\newtheorem{remark}[theorem]{Remark}

\numberwithin{equation}{section}

\title[Sharp Velocities for Periodic Operators]{Sharp Polynomial Velocity Decay Bounds for Multidimensional Periodic Schr\"odinger Operators}

\author[H.\ Abdul-Rahman]{Houssam Abdul-Rahman}
\address{[H.\ Abdul-Rahman] Department of Mathematical Sciences, United Arab Emirates University, Al Ain, UAE}
\email{\href{mailto:houssam.a@uaeu.ac.ae}{houssam.a@uaeu.ac.ae}}
	
\author[J.\ Fillman]{Jake Fillman}
\address{[J.\ Fillman] Department of Mathematics, Texas A\&M University, College Station, TX 77843-3368, USA}
\email{\href{mailto:fillman@tamu.edu}{fillman@tamu.edu}}

\author[C. Fischbacher]{Christoph Fischbacher}
\address{[C. Fischbacher] Department of Mathematics, Baylor University, Sid Richardson Bldg, 1410 S. 4th Street, Waco, TX 76706, USA }
\email{\href{mailto:c\_fischbacher@baylor.edu}{c\_fischbacher@baylor.edu}}

\author[W. Liu]{Wencai Liu}
\address{[W. Liu], Department of Mathematics, Texas A\&M University, College Station, TX 77843-3368, USA}
\email{\href{mailto:wencail@tamu.edu}{wencail@tamu.edu}}
\date{}

\thanks{H. A.-R. was supported in part by the UAE University under grant number G00004622.}

\thanks{J.\ F.\ was supported in part by the National Science Foundation under grant DMS--2513006 and by the  Simons Foundation under grant SFI-MPS-TSM-00013720.}
\thanks{C.\ F.\ was supported by the National Science Foundation under grant DMS--2510063.}
\thanks{W.\ L.\ was supported in part by the  National Science Foundation under grants  DMS-2246031  and DMS-2052572, by a Simons Fellowship in Mathematics, and by a Visiting Miller Professorship from the Miller Institute for Basic Research in Science, University of California, Berkeley.}

\begin{document}

\begin{abstract}
We investigate periodic Schr\"odinger operators in arbitrary dimensions in the large coupling regime. Our results establish that both the Lieb--Robinson velocity and the asymptotic velocity decay at an inverse polynomial rate in the coupling, with the precise exponent determined by the period of the underlying potential. In particular, we obtain sharp polynomial decay rates that capture the precise dependence on the periodic structure. 
\end{abstract}

\maketitle


\hypersetup{
	linkcolor={black!30!blue},
	citecolor={red},
	urlcolor={black!30!blue}
}

\section{Introduction}
Periodic quantum systems, modeled by Schr\"odinger operators with periodic potentials, are typically associated with ballistic transport, where wave packets spread linearly in time due to the underlying translational symmetry. These operators rigorously model quantum dynamics in crystalline and structured media, capturing the effects of translational symmetry on spectral and transport properties. They provide a natural setting for analyzing the relationship between spectral types and dynamical behavior. Surprisingly, recent work has revealed that while periodic systems do exhibit ballistic transport, the associated transport speed can be asymptotically small. This raises new questions about the subtle interplay between periodicity, spectral structure, and propagation speed.

The mathematical study of periodic Schr\"odinger operators has a long and rich history, tracing back to the foundational works of Bloch and Floquet in the early 20th century \cite{Bloch1929, Floquet1883}. These contributions established that solutions to the time-independent Schr\"odinger equation with periodic potentials can be expressed in terms of Bloch waves, leading to the band theory of solids and the identification of spectral bands and gaps \cite{AshcroftMermin}. From a mathematical perspective, periodic Schr\"odinger operators became central objects in spectral theory, particularly following the mid-20th century, with key developments such as the characterization of absolutely continuous spectra and the formulation of Floquet–Bloch theory, which reduces the spectral analysis to a direct integral over quasi-momentum space \cite{ReedSimonIV, Eastham1973}. In recent decades, attention has expanded beyond static spectral properties to include dynamical behavior, such as quantum transport and wavepacket spreading, with periodic systems often exhibiting ballistic dynamics and dipersive spreading under suitable conditions \cite{AschKnauf1998Nonlin, AVWW2011JMP, BoutetSabri2023OTAA, DFO2016JMPA, DFY2025disperse, DLY2015CMP, Fillman2021OTAA, SKW2025disperse}. These studies have deepened our understanding of the mathematical structure underlying quantum systems and have had significant impact on theoretical and applied physics, particularly in the design and analysis of photonic crystals and metamaterials \cite{Kuchment2001, Joannopoulos2008}. We refer the reader to the survey \cite{Kuchment2016BAMS} for a comprehensive overview and additional references and to \cite{DMY202XJST} for additional details regarding ballistic motion.  More recently, interdisciplinary approaches combining spectral theory with techniques from algebra and combinatorics have led to breakthroughs beyond what could be achieved through spectral methods alone. For example, some studies have focused on the irreducibility of associated algebraic varieties and its spectral consequences \cite{FLM2022JFA, Fillman2024JFA, FisLiShi2021CMP, LiShipman2020LMP, Liu2022GAFA, Shipman2020JST}.

While periodic Schr\"odinger operators are classically associated with ballistic transport and absolutely continuous spectrum, recent results have revealed a more nuanced picture. In particular, it has been shown that although wave packets in periodic systems do spread linearly in time, the rate of this spreading (the group velocity) can be made arbitrarily small. This phenomenon, first rigorously demonstrated in one-dimensional settings \cite{ADFS}, shows that the group velocity can become vanishingly small as the variations in a non-degenerate periodic potential are made increasingly large. A similar phenomenon has also been observed in quantum walks in periodic fields, as shown in the recent work \cite{AbdulRahmanStolz2023CMP}, where it was proved that the propagation velocity remains positive yet can be exponentially suppressed by carefully tuning the periodic structure. 
Moreover, existing proofs often rely heavily on fairly delicate one-dimensional techniques, particularly in the analysis of Lieb-Robinson-type velocity bounds, so the analysis of this phenomenon in higher dimensions was not understood and would indeed require substantial novel ideas and insights. 
This raises the question of whether more general or structurally simpler arguments can reveal the same effect and offer broader insight into the relationship between spectral structure and transport speed in periodic media.

In this work, we address the higher-dimensional setting by employing a novel synthesis of complex analysis and perturbation theory to eschew the delicate one-dimensional arguments used previously. 
Specifically, we show that for non-degenerate periodic potentials, the propagation velocity can be made arbitrarily small by increasing the amplitude of the potential, and we give a \emph{sharp} rate of decay.
This approach not only broadens the dimensional scope of the phenomenon but \emph{also} gives sharp estimates on the velocities and \emph{furthermore} even yields a simpler and more transparent proof in the one-dimensional setting.

The remainder of the paper is organized as follows.
In Section~\ref{sec:result}, we recall some terminology and formulate our main results precisely.
We use Rayleigh--Schr\"odinger perturbation series in Section~\ref{sec:rayleigh} to derive some useful estimates on eigenvalues, which we then employ in Section~\ref{sec:proof} to prove the main results.


\subsection*{Acknowledgements} The authors are grateful to G\"unter Stolz for helpful discussions.
   W.L.  thanks the Department of Mathematics at UC Berkeley for its hospitality where part of this work were done during his visits in Fall 2024 and Fall 2025.

\section{Setting and Main Results} \label{sec:result}
We consider $V: \bbZ^d \to \bbR$ and the operator $H = \Delta + V$ on $\ell^2(\bbZ^d)$ given by
\[
[H\psi](n)
= V(n) \psi(n) + \sum_{|n-m|_1=1} \psi(m), \quad n \in \bbZ^d,
\]
for which the function $V:\bbZ^d \to \bbR$ is \emph{periodic}.
Given $p = (p_1,\ldots,p_d) \in \bbN^d$, we denote the period lattice by
$p\bbZ^d := p_1 \bbZ \oplus \cdots \oplus p_d \bbZ$.
We say that $V$ is $p$-periodic\footnote{One could say that $V$ is periodic if there exists some full-rank lattice $\Gamma \subseteq \bbZ^d$ such that $V(\cdot-\gamma)=V$ for every $\gamma \in \Gamma$. 
However, it is not hard to check that any such $\Gamma$ contains a subgroup of the form $p \bbZ^d$ for suitable $p$ and hence no generality is lost.}  if $V(n+m)=V(n)$ for all $n \in \bbZ^d$ and $m \in p \bbZ^d$.
We define the fundamental cell by putting $[\ell] = \{0,1,\ldots, \ell -1\}$ for $\ell \in \bbN$ and
\begin{eqnarray}
    W = [p_1] \times [p_2] \times \cdots\times [p_d],
\end{eqnarray}
and we say that $V$ is \emph{non-degenerate} if $V$ is injective on $W$.
The \emph{separation} of $V$ is defined by
\begin{equation}
    \separate V 
    = \min\{|V(n) - V(m) |: n,m \in W, \ n \neq m \}.
\end{equation}
One has $\separate (\mu V) = |\mu| \separate (V)$ for each $\mu \in \bbR$ and we see that $V$ is non-degenerate if and only if $\separate(V) > 0$.

Our main results concern the unitary evolution $e^{-{\iop}tH}$, which is known to propagate \emph{ballistically} in the sense that the evolution of the position operator, $X$, satisfies:
\begin{equation}
    \tfrac{1}{t}\underbrace{e^{{\iop}tH} X e^{-{\iop}tH}}_{=:X(t)} \xrightarrow{\ \mathrm{s} \ } G,
\end{equation}
with $G$ a self-adjoint operator having trivial kernel.
This was proved first for continuum Schr\"odinger operators \cite{AschKnauf1998Nonlin}, and later generalized to many other systems, including discrete Schr\"odinger operators, quantum walks, and more general periodic graph operators \cite{AVWW2011JMP, BoutetSabri2023OTAA, DFO2016JMPA, DLY2015CMP, Fillman2021OTAA}.

To study transport in the strong coupling regime, we introduce the scaled periodic Schrödinger 
\begin{equation}
H_\mu := \Delta + \mu V,
\end{equation}
where the coupling parameter $\mu>0$ is taken to be large.  Note that the introduction of a large parameter $\mu$ amplifies the variations in the potential landscape. Physically, this corresponds to placing increasingly high barriers between sites of different potential values. As $\mu$ increases, tunneling between distinct potential sites is suppressed. Consequently, the propagation of the quantum state slows down, and in the limit $\mu \to \infty$, one expects a strong suppression of transport.

One measure of the rate of spreading is quantified by the \emph{asymptotic velocity}.
For a state $\psi \in D(X)$, we denote
\begin{eqnarray} \label{eq:vAsyDef}
    v_\asy(H,\psi) := \limsup_{t \to \infty} \frac{1}{t} \|X e^{-\iop t H} \psi\|,
\end{eqnarray}
where $X:D(X) \subseteq \ell^2(\bbZ^d) \to (\ell^2(\bbZ^d))^d$ denotes the (vector-valued) position operator
\begin{eqnarray}
    X\psi = (X_1\psi,\ldots, X_d\psi), \quad [X_i \psi](x)  = x_i \psi(x).
\end{eqnarray}
We then define the \emph{asymptotic velocity} of $H$ to be
\begin{equation}\label{def:v-asy}
    v_\asy(H)=
    \sup_{\substack{\psi \in D(X) \\ \|\psi\|=1}} v_\asy(H,\psi).
\end{equation}
We use here the refined version \eqref{def:v-asy}, involving the supremum over normalized states in the domain of $X$, which was first considered in an ongoing work on the velocity of unitary models (quantum walks and CMV matrices) \cite{ARCSW-QW-CMV}.

Here we note that when $H$ is periodic, the limit of the quantity on the right-hand-side of \eqref{eq:vAsyDef} exists so one can replace $\limsup$ by $\lim$ in the setting in which we work.

Our main result shows a \emph{sharp} rate of decay for $v_\asy$ in the regime of large coupling constant.

\begin{theorem} \label{t:asymptotic}
Suppose $V:\bbZ^d \to \bbR$ is $p$-periodic and non-degenerate with $p_0 := \min\{p_i : i=1,2,\ldots,d\}$.
Then, as $\mu \to \infty$, one has
\begin{equation} \label{eq:vAsyBounds}
     v_\asy(H_\mu)= C\mu^{-p_0+1} + O(\mu^{-p_0}),
\end{equation}
where $C>0$ is a constant depending on $V$ and $p$. Moreover,
\begin{equation} \label{eq:vAsydeltaBounds}
     v_\asy(H_\mu,\delta_0) = c\mu^{-p_0+1} + O(\mu^{-p_0})
\end{equation}
for a constant $c>0$ depending on $V$ and $p$.
\end{theorem}
In fact, our argument gives more information than is stated in the theorem in the sense that the asymptotic velocity in each coordinate direction is quantified explicitly as a function of the period in that direction (cf.\ \eqref{eq:GinormExpansion}). 
Broadly speaking, transport is suppressed the most in directions with larger periods.
For example, in a two dimensional system, if $p_1=1$ and $p_2>1$, then the potential is constant along each line parallel to the horizontal axis, creating a propagating channel along which the wave packet freely travels, leading to the lack of decay of $v_\asy$ seen above. However, the transport in the vertical axis is suppressed in the sense that
$\lim X_2(t)/t$ has order $\mu^{-p_2+1}$.

{A related estimate was obtained in the one-dimensional setting 
in~\cite{ADFS}, where an upper bound of the form 
$v_{\mathrm{asy}}(H_\mu) \leq C \mu^{-p+1}$ was established (where $p$ is the period of the potential). 
Theorem \ref{t:asymptotic} provides not only a bound above and below but a sharp polynomial (in $\mu$) asymptotic statement in \emph{arbitrary} dimension.}

One can also quantify the notion of propagation (transport) in quantum systems using a single-body analogue of the \emph{Lieb-Robinson velocity} \cite{LiebRobinson1972CMP}, of the form
\begin{equation} \label{eq:genericLRbound}
    |\langle \delta_n, e^{-\iop tH_\mu} \delta_m \rangle| \lesssim e^{-\rho_0 (|n - m|_1 - v_{\LR} |t|)}, \quad \text{for all } n,m \in \bbZ^d, \ t \geq 0,
\end{equation}
for suitable constants $\rho_0, v_{\LR} > 0$. Here and in the following, $\{\delta_n : n\in\bbZ^d\}$ denotes the canonical basis of $\ell^2(\bbZ^d)$. Inequality (\ref{eq:genericLRbound}) implies that $v_{\LR}$ serves as an effective upper bound on the speed of information propagation, up to exponentially decaying corrections. 
The physical interpretation is as follows \cite{AzeimanWarzel-JMP12}: starting from the localized initial state $\delta_m$, the probability that the distance traveled exceeds $v|t|$ at time $t$, for any velocity $v > 0$, is given by
\begin{eqnarray}
\sum_{n:\, |n - m|_1 > v|t|} |\langle \delta_n, e^{-\iop tH_\mu} \delta_m \rangle|^2 
&\lesssim& 
e^{-2\rho |t| (v - v_{\LR})}\sum_{n:\, |n - m|_1 > v|t|} e^{-2\rho (|n - m|_1 - v |t|)} \notag\\
&\lesssim& e^{-2\rho |t| (v - v_{\LR})}.
\end{eqnarray}
Thus, for any $v > v_{\LR}$, the probability of observing propagation beyond the linear bound $|n - m|_1 = v|t|$ decays exponentially in time. This defines a so-called \emph{Lieb-Robinson light cone} with effective velocity $v_{\LR}$, see e.g., \cite{Arbunich2023CPDE, Breteaux2023arxiv, CJWWpreprint, Tran2021PRL}.

Our second main result is a bound on the Lieb-Robinson velocity $v_\LR$ decaying at the \emph{sharp} rate $\mu^{-p_0+1}$ as $\mu \to \infty$, with $p_0:=\min\{p_1,\ldots,p_d\}$.

\begin{theorem} \label{t:LiebRob}
Suppose $V:\bbZ^d \to \bbR$ is $p$-periodic and non-degenerate with $p_0 := \min\{p_i: i=1,2,\ldots,d\}$.
Then, for \emph{any} $\rho_0>0$, there are constants $C,C_1,\mu_0>0$ depending on $d$, $p$, $\separate V$, and $\rho_0$ such that \eqref{eq:genericLRbound} holds with $v_\LR = C_1\mu^{-p_0+1}$. That is, for all $n,m \in \bbZ^d$ and all $\mu \geq \mu_0$, one has
        \begin{equation} \label{eq:mainBound}
        | \langle \delta_n, e^{-{\iop}tH_\mu} \delta_m \rangle|
        \leq    C e^{-\rho_0(|n-m|_1-C_1 \mu^{-p_0+1}|t|)}
        \end{equation}
for all $t \geq 0$.

Furthermore, this is optimal in the sense that if a bound of the form \eqref{eq:genericLRbound} holds for all $\mu$ large, then $v_\LR \geq c\mu^{-p_0+1}$ for a constant $c>0$.
    \end{theorem}

Theorem \ref{t:LiebRob} shows that the Lieb--Robinson velocity decays at the sharp rate 
$\mu^{-p_{0}+1}$ in the strong coupling regime. 
In comparison with earlier work in~\cite{ADFS},  
Theorem \ref{t:LiebRob} identifies the precise rate of decay and applies in higher dimensions, both of which appear to lie beyond the scope of \cite{ADFS}. Thus, our result both 
broadens the scope and sharpens the known velocity bounds.
We emphasize that Theorem \ref{t:LiebRob} yields a polynomial rate of decay in the 
coupling constant rather than an exponential one. The threshold 
$\mu_{0}$ can be chosen independently of the period of the potential, 
and in fact one may take $\mu_{0}=1/\varepsilon_{0}$, where $\varepsilon_{0}=\varepsilon_0(d,\separate V,\rho_0)$ given in Lemma \ref{lem:analyticity} below.

We also note that both the Lieb--Robinson velocity $v_{\mathrm{LR}}$ and 
the asymptotic velocity $v_{\mathrm{asy}}$ exhibit the same scaling 
behavior with respect to the coupling constant, namely $\mu^{-p_{0}+1}$. 
The same phenomenon has been observed in the context of quantum walks \cite{ARCSW-QW-CMV}, 
where these two velocities also scale in the same way. This raises the 
natural question of whether there exist settings in which the two notions 
of velocity exhibit different scaling behavior. 

\begin{remark}
    It can be seen from the proofs of the main results that the argument is fundamentally \emph{graph} theoretic in nature, so it generalizes readily to the case of periodic Schr\"odinger operators on $\bbZ^d$-periodic graphs.
    If $\mathcal{G} = (\mathcal{V}, \mathcal{E})$ is $\bbZ^d$-periodic and $V:\mathcal{V}  \to \bbR$ is $p$-periodic and non-degenerate, then the same proofs as in the case $\mathcal{G} = \bbZ^d$ produce similar results for the operator $\Delta + V$ (with $\Delta$ the graph Laplacian) with $p_0$ equal to the minimal combinatorial distance separating two distinct vertices in the same $p\bbZ^d$-orbit.
\end{remark}

Let us reiterate the strength of the results: we improve the rate of decay from \cite{ADFS}, obtain a sharp rate of decay, and generalize to higher dimensions, all with a conceptually simpler argument that can even be applied to $\bbZ^d$-periodic graphs other than the square lattice $\bbZ^d$ itself.

    \section{Asymptotics of Eigenvalues and Eigenvectors}
\label{sec:rayleigh}

We will need to derive some asymptotic statements for eigenvalues and eigenvectors of Floquet matrices.
    Throughout the discussion, we fix a nonempty finite set\footnote{It is convenient to use a general index set, such as the fundamental cell  in the main application} $S$ and consider $\bbC^{S \times S}$, the collection of $S \times S$ matrices with complex entries.
    Let $\{e_n : n \in S\}$ be the standard basis of $\bbC^S$.
    In the discussion below, one can picture $B$ as the Floquet operator of a Laplacian on a suitable fundamental domain and $D$ as the associated potential.
    We note in particular that it is not needed to assume the operators are Hermitian, so this can be applied to \emph{complex} potentials and \emph{complex quasimomenta}.

\begin{theorem} \label{t:rayleigh}
Assume $S$ is a nonempty set $B,D \in \bbC^{S\times S}$ are matrices such that:
\begin{itemize}
    \item \( D = \mathrm{diag}(D_n) \) with  pairwise distinct entries,
    \item $B_{nn}=0$ for all $n \in S$.
\end{itemize}

Let \( A_\varepsilon = D + \varepsilon B \).
For $\varepsilon \in \bbR$ sufficiently small, the  eigenvalues and eigenvectors have convergent expansions of the form
\begin{align*}
\eta_n(\varepsilon) = \sum_{r=0}^\infty \eta_n^{(r)} \varepsilon^r, \qquad
u_n(\varepsilon) = \sum_{r=0}^\infty u_n^{(r)}\varepsilon^r,
\end{align*}
respectively,
where \( \langle e_n, u_n^{(r)} \rangle = 0 \) for all \( r \geq 1 \), 
\begin{align} \label{eq:etaUr=0}
    \eta_n^{(0)}  = D_n, \quad u_n^{(0)} = e_n,
\end{align} 
and
\begin{align} \label{eq:etaUr=1}
        \eta_n^{(1)}  = 0, \quad (u_n^{(1)})_m = \frac{B_{mn}}{D_n-D_m}, \quad m \neq n.
\end{align}
Furthermore, for all $r \geq 2$:
    \begin{align} \label{eq:mu(n)inductive}
        \eta_n^{(r)} 
        &=     \sum_{m  \neq n} B_{nm}(u_n^{(r-1)})_m \\ 
        \label{eq:v(n)inductive}
        \left(u_n^{(r)} \right)_m
        &= \frac{1}{D_n-D_m}
        \left(  \sum_{\ell \neq n} B_{m\ell} (u_n^{(r-1)})_\ell - \sum_{s=2}^{r-1} \eta_n^{(s)} (u_n^{(r-s)})_m\right), \quad m \neq n.
    \end{align}
\end{theorem}

\begin{remark}
    Note that the sum in \eqref{eq:mu(n)inductive} starts from $s=2$ and in particular is absent for $r=2$. 
\end{remark}

\begin{proof}[Proof of Theorem~\ref{t:rayleigh}]
This follows from the Rayleigh--Schr\"odinger expansion. To keep the paper more self-contained, we give the details here.
By definition, $A_\varepsilon$ is an analytic function of $\varepsilon$. Since $D$ has simple spectrum, the existence of an expansion follows from perturbation theory (compare e.g.\ \cite[Section~II.1.1]{Kato}).
We thus fix $n \in S$, $  \eta_n(0) = D_n$, and $u_n(0) = e_n$, and consider their analytic continuations (note that this immediately gives us \eqref{eq:etaUr=0}).
Since any analytic vector-valued function has analytic components, we can always normalize by the $n$th component to obtain an expansion for $u_n$ such that $(u_n(\varepsilon))_n \equiv 1$, whence $(u_n^{(r)})_n =0$ for all $r \geq 1$.

Expanding $A_\varepsilon u_n(\varepsilon) = \eta_n(\varepsilon)u_n(\varepsilon)$ and collecting terms gives
\begin{equation} \label{eq:seriesCollected}
    Du_n^{(0)}
    + \sum_{r=1}^\infty (Du_n^{(r)} + Bu_n^{(r-1)})\varepsilon^r
    = \eta_n^{(0)}u_n^{(0)} + \sum_{r=1}^\infty \sum_{s=0}^r \eta_n^{(s)} u_n^{(r-s)} \varepsilon^r.
\end{equation}
Since $D$ is diagonal and $\langle e_n, u_n^{(r)} \rangle = 0$ for $r \geq 1$, we have 
\[ \langle e_n, Du_n^{(r)} \rangle = 0 \]
for all $r \geq 1$.  
Using this, project both sides of \eqref{eq:seriesCollected} onto $e_n$ to get:
\begin{equation}
    D_n + \sum_{r=1}^\infty \langle e_n, B u_n^{(r-1)} \rangle \varepsilon^r
    = \eta_n^{(0)} + \sum_{r=1}^\infty \eta_n^{(r)} \varepsilon^r,
\end{equation}
giving\footnote{Here let us remark that, formally, the ``$m \neq n$'' on the summation is not required, since $B_{nn}=0$ by assumption. However, it is a useful reminder, so we leave it in the notation.}
\begin{equation}
    \eta_n^{(r)} 
    = \langle e_n, B u_n^{(r-1)} \rangle
    = \sum_{m \neq n} B_{nm}(u_n^{(r-1)})_m,
\end{equation}
which proves \eqref{eq:mu(n)inductive}. Note that $\eta_n^{(1)} = 0$ follows directly.\footnote{Let us point out this also follows from the Feynman--Hellmann theorem directly.}

Projecting \eqref{eq:seriesCollected} onto $e_m$ with $m \neq n$ yields
\begin{equation}
    D_m (u_n^{(r)})_m + (Bu_n^{(r-1)})_m = \sum_{s=0}^{r} \eta_n^{(s)} (u_n^{(r-s)})_m,
\end{equation}
which, with the help of $u_n^{(0)} = e_n$, $\eta_n^{(0)}=D_n$, becomes:
\begin{equation}
    (D_m-D_n) (u_n^{(r)})_m + (Bu_n^{(r-1)})_m 
    = \sum_{s=1}^{r-1} \eta_n^{(s)} (u_n^{(r-s)})_m,
\end{equation}
which proves \eqref{eq:v(n)inductive} (and the second identity in \eqref{eq:etaUr=1}) after rearranging and recalling that $\eta_n^{(1)}=0$. 
\end{proof}

Let us compute a few terms in these expansions by hand and then establish a general pattern.
As already noted, we have
\begin{equation}
    \eta_n^{(0)}=D_n, \quad u_n^{(0)} = e_n, \quad \eta_n^{(1)}=0, \quad
    (u_n^{(1)})_m = \frac{B_{mn}}{D_n - D_m}, \quad m \neq n.
\end{equation}
In general, we see from Theorem~\ref{t:rayleigh} that $u_n^{(r)}$ is determined by $\{\eta_n^{(s)} : s<r\}$ and by the vectors $\{u_n^{(s)} : s<r\}$, (as well as the requirement $\langle e_n, u_n^{(r)} \rangle = 0$ for $r\geq 1$) whereas $\eta_n^{(r)}$ is determined by $u_n^{(r-1)}$.

The next coefficients of $\eta$ and $u$ can then be computed as
\begin{equation} \label{eq:eta2expression}
    \eta_n^{(2)} = \sum_{m  \neq n} B_{nm} (u_n^{(1)})_m 
    = \sum_{m \neq n} \frac{B_{nm}B_{mn}}{D_n-D_m}.
\end{equation}
and
\begin{align}
    \nonumber\left(u_n^{(2)} \right)_m
        & = \frac{1}{D_n - D_m}\sum_{m_1 \neq n} B_{mm_1}(u_n^{(1)})_{m_1}  \\
         & =   \sum_{m_1 \neq n} \frac{ B_{mm_1} B_{m_1n}}{(D_n - D_m)(D_n - D_{m_1})},
\end{align}
which allows us to compute the third-order coefficients:
\begin{equation} \label{eq:eta3expression}
    \eta_n^{(3)} = \sum_{m  \neq n} B_{nm} (u_n^{(2)})_m 
    = \sum_{\substack{m \neq n \\ m_1  \neq n} }   \frac{ B_{nm} B_{mm_1} B_{m_1n}}{(D_n - D_m)(D_n - D_{m_1})}
\end{equation}
and
\begin{align}
    \nonumber\left(u_n^{(3)} \right)_m
        & = \frac{1}{D_n - D_m}\sum_{m_1 \neq n} B_{mm_1}(u_n^{(2)})_{m_1}  - \frac{1}{D_n-D_m}\left( \sum_{m_1  \neq n} \frac{B_{nm_1}B_{m_1n}}{D_n-D_{m_1}} \right)\left( \frac{B_{mn}}{D_n - D_m} \right) \\
         & = \sum_{m_1,m_2 \neq n} \frac{ B_{mm_1} B_{m_1m_2} B_{m_2n} }{(D_n - D_{m})(D_n - D_{m_1})(D_n - D_{m_2})} - \sum_{m_1 \neq n} \frac{B_{mn} B_{nm_1} B_{m_1n}}{(D_n-D_m)^2(D_n - D_{m_1})}.
\end{align}

We can give this a useful interpretation by thinking of $B$ as an operator on a (directed) graph $\Gamma$ with vertex set $S$ and an edge from $n$ to $m$ whenever $B_{nm} \neq 0$.
We then define a  \emph{path} of \emph{length} $\ell$ from $m$ to $n$ (notation: $\gamma:m\to n$) in $\Gamma$ to be a finite sequence $\gamma = (n_0,\ldots,n_\ell)$ where $n_0=m$, $n_\ell=n$, and $(n_{s-1},n_s)$ is an edge for every $1 \le s \le \ell$. 
Write $|\gamma|:=\ell$ for the length of $\gamma$.
We say that $\gamma$ is \emph{irreducible} if $n_s \neq n$ for all $s=1,2,\ldots,\ell-1$.
For a path $\gamma=(n_0,\ldots, n_\ell)$, we denote $B_\gamma = B_{n_0,n_1} \cdots B_{n_{\ell-1},n_\ell}$.
We call $\gamma$ a \emph{loop} if $n_0 = n_\ell$.

With this interpretation, we see that $\eta_n^{(r)}$ has the form
\begin{equation} \label{eq:etaloopExpansion}
    \eta_n^{(r)} = \sum_{\substack{\gamma:n \to n \\ |\gamma|=r}} f(\gamma,D)B_\gamma,
\end{equation}
where the notation indicates that the sum runs over loops of length $r$ from $n$ to itself, and the coefficient $f(\gamma,D)$ depends on the loop $\gamma$ and the values assumed by $D$ on vertices in the loop.
Likewise, $u_n^{(r)}$ has the form
\begin{equation}\label{eq:uloopExpansion}
    (u_n^{(r)})_m   =  \sum_{\substack{\gamma:m \to n \\ |\gamma|=r}} g(\gamma,D)B_\gamma,
\end{equation}
where again the notation indicates that the sum runs over paths from $m$ to $n$ of length $r$ and once again the coefficient depends on the path $\gamma$ and the values assumed by $D$ on the path.

\begin{coro} \label{coro:loopExpansion}
    With assumptions as above,   $\eta_n^{(r)}$ has an expansion of the form \eqref{eq:etaloopExpansion} and $u_n^{(r)}$ has an expansion of the form \eqref{eq:uloopExpansion} for every $n \in S$ and every $r \geq 1$.
    
    For an {irreducible} loop $\gamma:n \to n$, $\gamma = (n_0,\ldots,n_r)$, we have
    \begin{equation} \label{eq:fgamma:irrLoop}
        f(\gamma,D) = \prod_{j=1}^{r-1} \frac{1}{D_n-D_{n_j}}.
    \end{equation}
    Likewise, for $m\neq n$ and an {irreducible} path $\gamma:m \to n$, $\gamma = (n_0,\ldots,n_r)$, we have
    \begin{equation} \label{eq:fgamma:irrPath}
        g(\gamma,D) = \prod_{j=0}^{r-1} \frac{1}{D_n-D_{n_j}}.
    \end{equation}
\end{coro}

\begin{proof}
    The proof is by induction on $r$. The desired statement (including \eqref{eq:fgamma:irrLoop} and \eqref{eq:fgamma:irrPath}) is already established for $r =1,2,3$ in the discussion above.
    Inductively, assume the statement of the theorem holds up to order $r-1$.

    If $\gamma = (n_0,\ldots,n_\ell)$ and $\gamma' = (m_0,\ldots,m_{\ell'})$ are paths with $n_\ell = m_0$, we write $\gamma \circledast \gamma' =  (n_0,\ldots,n_\ell,m_1,\ldots,m_{\ell'})$ for their \emph{amalgamation}. 
    Notice that
\begin{equation}
    B_{\gamma \circledast \gamma'} = B_\gamma B_{\gamma'}.
\end{equation}
We then have (with the help of the inductive hypothesis)
\begin{align*}
    \eta_n^{(r)} 
        =     \sum_{m} B_{nm}(u_n^{(r-1)})_m 
        & = \sum_m B_{nm} \sum_{\substack{\gamma:m \to n \\ |\gamma| = r-1}} g(\gamma,D)B_\gamma \\
        & = \sum_m \sum_{\substack{\gamma:m \to n \\ |\gamma| = r-1}} g(\gamma,D) B_{(n,m) \circledast \gamma}.
\end{align*}
For $\gamma$ a path of length $r-1$ from $m$ to $n$, we see that $(n,m) \circledast \gamma$ is a loop from $n$ to $n$ of length $r$, which is irreducible if and only if $\gamma$ is irreducible.
Thus, the inductive step for $\eta$ and \eqref{eq:fgamma:irrLoop} is concluded.

Similarly, we have
\begin{align*}
        \left(u_n^{(r)} \right)_m
        & = 
        \frac{1}{D_n - D_m} \left( \sum_{m_1
        }B_{mm_1}(u_n^{(r-1)})_{m_1}  
        - \sum_{s=2}^{r-1} \eta_n^{(s)} (u_n^{(r-s)})_m\right) \\
        & = 
        \frac{1}{D_n - D_m} \left( \sum_{m_1} B_{mm_1} \sum_{\substack{\gamma:m_1\to n \\ |\gamma|=r-1}} g(\gamma,D) B_\gamma
        - \sum_{s=2}^{r-1} 
        \left[
        \sum_{\substack{\gamma:n\to n \\ |\gamma|=s}} f(\gamma,D)B_\gamma\right]
        \left[\sum_{\substack{\gamma':m\to n \\ |\gamma'|=r-s}} g(\gamma',D)B_{\gamma'}\right]\right) \\
        & = 
        \sum_{m_1} \sum_{\substack{\gamma:m_1\to n \\ |\gamma|=r-1}} \frac{g(\gamma,D)}{D_n - D_m}  B_{(m,m_1) \circledast \gamma}
        - \sum_{s=2}^{r-1} 
        \sum_{\substack{\gamma:n\to n \\ |\gamma|=s}} 
        \sum_{\substack{\gamma':m\to n \\ |\gamma'|=r-s}} 
        \frac{f(\gamma,D) g(\gamma',D)}{D_n - D_m}B_{\gamma'\circledast \gamma}.
\end{align*}
In the first term, we observe that $(m,m_1)\circledast \gamma$ is a path of length $r$ from $m$ to $n$ (which is irreducible if and only if $\gamma$ is irreducible), while,
in the second term,  $\gamma'\circledast\gamma$ is a path from $m$ to $n$ of length $r$ (and $\gamma' \circledast\gamma$ is never irreducible), concluding the inductive step for $u$ and \eqref{eq:fgamma:irrPath}.
\end{proof}

\begin{remark}
    These ideas and expansions are useful and have been employed in other contexts. As this paper was being completed, the preprint \cite{faust2025absenceflatbandsdiscrete} was posted, which, independent of us, employed a similar expansion in the study of flat bands.
\end{remark}

\section{Proof of Main Results}
\label{sec:proof}

\subsection{Floquet Theory}
Let us briefly review a few aspects of Floquet theory that we need. 
For more background, see \cite{Kuchment2016BAMS}  and for proofs of relevant statements that apply in the present setting, see \cite{FLM2022JFA}.

Recall that $W = [p_1] \times \cdots \times [p_d] \cong \bbZ^d / p \bbZ^d$ is the fundamental domain, where $[m]=\{0,1,\ldots,m-1\}$. 
Put
\begin{equation}
    P =\# W = p_1p_2 \cdots p_d.
\end{equation}
Given $z = (z_1, z_2, \ldots, z_d) \in (\bbC^*)^d$ (where $\bbC^* = \bbC \setminus \{0\}$), we write $H(z)$ for the restriction of $H$ to the space
\begin{equation}
\mathscr{H}_p  (z) :=    \{\psi \in \bbC^{\bbZ^d} : \psi(n+p_j e_j) = z_j \psi(n)\}.
\end{equation}

The matrix $H(z)$ is self-adjoint and hence has real eigenvalues for any $z \in (\partial \bbD)^d$.
For such $z$, the spectral band functions $\lambda_j(z)$, $1 \le j \le P$, are then defined by listing the eigenvalues of $H(z)$ in order with  multiplicity:
\begin{equation}
    \spectrum H(z) = \{\lambda_j(z) : 1 \le j \le P \}, 
    \quad \lambda_1(z) \leq \lambda_2(z) \leq \cdots \leq \lambda_P(z).
\end{equation}
For each $z$ for which $H(z)$ is diagonalizable, there is an invertible similarity transformation $Q(z)$ such that
\begin{equation}
    [Q(z)]^{-1} H(z)Q(z) = \diag(\lambda_j(z)).
\end{equation}

In this work, we are interested in the behavior of $H_\mu = \Delta+ \mu V$ as $|\mu| \to \infty$.
Rescaling by $\mu$, it suffices to consider $A_\varepsilon := \varepsilon \Delta + V$ as $|\varepsilon| \to 0$.
Let us write $A_\varepsilon(z)$ for the corresponding Floquet matrices.

We also define for $\rho>0$ the set
\begin{equation}
    \annulus(\rho)
    := \{ z \in (\bbC^*)^d : |\log|z_j||<\rho \text{ for all } 1\le j \le d\}.
\end{equation}

\begin{lemma} \label{lem:analyticity}
    Assume $V:\bbZ^d \to \bbR$ is $p$-periodic and  non-degenerate, let $\rho_0>0$ be given, and put
    \begin{equation}
        \varepsilon_0 = \varepsilon_0(d,\separate V,\rho_0) 
        := \frac{\separate V}{8d(1+\cosh (2\rho_0))}.
    \end{equation}

For each $|\varepsilon| <  2 \varepsilon_0$, the spectrum of $A_\varepsilon(z)$ is simple for every $z \in \annulus(2\rho_0)$, and each $\eta_n(\varepsilon,z)$ is an analytic function of $(\varepsilon,z) \in (-2\varepsilon_0, 2\varepsilon_0) \times \annulus(2\rho_0)$ for each $n$.
\end{lemma}

\begin{proof}
If $|\varepsilon| < 2 \varepsilon_0$, simplicity follows from the Gershgorin circle theorem,
so the desired analyticity follows from eigenvalue perturbation theory \cite{Kato}.
\end{proof}

Choosing $|\varepsilon|$ small enough that Lemma~\ref{lem:analyticity} applies and $z \in \annulus(2\rho_0)$, we define  $\eta_n(\varepsilon,z)$ to be the Floquet eigenvalue of $A_\varepsilon(z)$ analytically continued from $\eta_n(0,z)=V(n)$, and let $Q_\varepsilon(z)$ denote the matrices diagonalizing $A_\varepsilon(z)$:
\begin{equation} \label{eq:QepsilonConjugatesAepsilon}
    [Q_\varepsilon(z)]^{-1} A_\varepsilon(z)Q_\varepsilon(z) = \diag(\eta_n(\varepsilon,z)).
\end{equation}
We choose $Q_\varepsilon(z)$ by selecting the columns to be the analytic eigenvectors as in Theorem~\ref{t:rayleigh}.

With notation  as above, we define
\begin{equation}
    L^2(\bbT^d, \bbC^W; \tfrac{\rmd\theta}{|\bbT^d|})
    =
    \left\{ f:\bbT^d \to \bbC^W \ \vert \  \|f\| := \left[ \int_{\bbT^d} \|f(\theta)\|_{\bbC^W}^2 \, \frac{\rmd\theta}{|\bbT^d|} \right]^{1/2} < \infty \right\},
\end{equation}
where $\bbT:=\bbR/(2\pi \bbZ)$ and we write $\rmd\theta = \rmd\theta_1 \, \cdots \rmd\theta_d$ for the standard Lebesgue measure on $\bbT^d$.
The Floquet transform $\floquet: \ell^2(\bbZ^d) \to L^2(\bbT^d, \bbC^W; \frac{\rmd\theta}{|\bbT^d|})$ is then given by
\begin{equation}
    \delta_{n+x \odot p} \mapsto 
    e^{-{\iop} \langle x, \cdot \rangle}e_n, \quad n \in W, \ x \in \bbZ^d, 
\end{equation}
where we  define 
\[
x \odot p = (x_1p_1,\ldots,x_dp_d), \quad \langle x, \theta \rangle = \sum_{j=1}^d x_j\theta_j.\]
By direct computations, one can see that $\floquet$ is unitary and conjugates $A_\varepsilon$ to a decomposable operator, viz.:
\begin{equation}
    (\floquet A_\varepsilon \floquet^{-1} g)(\theta) 
    = A_\varepsilon(e^{{\iop}\theta})g(\theta), 
    \quad g \in L^2(\bbT^d,\bbC^W; \tfrac{\rmd\theta}{|\bbT^d|}),
\end{equation}
where $e^{\iop \theta} =(e^{\iop \theta_1}, \ldots, e^{\iop \theta_d})$.
For $x,y \in \bbZ^d$ and $n,m \in W$ given, we note by unitarity of $\floquet$ that
\begin{align}
\nonumber
    \langle \delta_{m+y \odot p}, e^{-{\iop}tA_\varepsilon} \delta_{n+x \odot p} \rangle
    & = \langle \floquet \delta_{m+y \odot p}, \floquet e^{-{\iop}t A_\varepsilon }\floquet ^{-1}\floquet  \delta_{n+x \odot p} \rangle \\
    \label{eq:FloquetBlockRepn}
    & = \int_{\bbT^d} \langle e_m, e^{-{\iop} t A_\varepsilon(e^{{\iop}\theta})}e_n \rangle e^{-{\iop} \langle (x-y),\theta \rangle} \, \frac{\rmd\theta}{|\bbT^d|}.
\end{align}
In our discussion of bounds on  the Lieb--Robinson velocity, the desired estimates are most conveniently formulated in terms of bounds on the \emph{blocks} of the unitary propagator with respect to the period lattice. 
Concretely (with the period $p$ fixed), we denote
\begin{equation}
    A(x,y) = \Pi(W+y\odot p) \, A \, \Pi(W+x \odot p)^*, \quad x,y \in \bbZ^d,
\end{equation}
where $\Pi(S)$ denotes the canonical projection $\ell^2(\bbZ^d) \to \ell^2(S)$.
We note that each such block is represented as a matrix in $\bbC^{W \times W}$.
With the help of \eqref{eq:FloquetBlockRepn}, we obtain the following representation of the blocks of $e^{-{\iop}t A_\varepsilon}$:
\begin{equation} \label{eq:duhamel}
    e^{- {\iop} t A_\varepsilon}(x,y) 
    = \int_{\bbT^d}   e^{-{\iop}t A_\varepsilon(e^{{\iop}\theta})} e^{-{\iop} \langle (x-y),\theta \rangle} \, \frac{\rmd\theta}{|\bbT^d|}.
\end{equation}

\begin{remark} \label{rem:blocksToEntries}
Using the triangle inequality and well-known inequalities for matrix norms, we can pass between estimates on blocks and estimates on matrix elements of the propagator $e^{-\iop t A_\varepsilon}$ at the expense of some explicit $p$-dependent constants.
Concretely, for $n \in W + x\odot p$ and $m \in W + y \odot p$, we have
\begin{equation}
    p_0|x-y|_1- C \leq |n-m|_1 \leq p_{\max} |x-y|_1 + C.
\end{equation}
where $C = 2\sum p_j$ and $p_{\max} = \max\{p_j\}$.
Likewise, with $\|\cdot\|$ the standard Euclidean matrix norm and $\|\cdot \|_\infty$ denoting the maximum modulus of an entry, we have the beloved bound
\begin{equation}
    \|B\|_\infty \leq \|B\| \leq P \|B\|_\infty.
\end{equation}
Taken together, these observations enable one to pass between estimates on blocks and matrix elements.
\end{remark}

\subsection{Eigenvalue Perturbation}

By the nondegeneracy assumption and Lemma~\ref{lem:analyticity}, the  Floquet eigenvalues of $A_\varepsilon$ have an expansion
\begin{equation}
    \eta_n(\varepsilon,z) = \sum_{r=0}^\infty \eta_n^{(r)}(z) \varepsilon^r
\end{equation}
for small $\varepsilon$. Let us begin by discussing the low-order expansion coefficients.
\begin{lemma} \label{lem:etaBounds}
    Assume $V$ is $p  $-periodic and non-degenerate, and put $p_0 = \min \{p_1,\ldots,p_d\}$.
    We have:
    \begin{enumerate}[label={\rm(\alph*)}]
        \item For $r=0,1,2,\ldots,p_j-1$, $\eta_n^{(r)}(z)$ is real and independent of $z_j$.
        In particular, if $r<p_0$, then $\eta_n^{(r)}(z)$ is real and independent of $z$.
        \item For $r= p_j$, $\eta_n^{(r)}(z)$ is of the form
        \begin{equation} \label{etajp_0:cosineform}
            \eta_n^{(r)}(z) = \sum_{j:p_j=p_0} c_{j,n}(z_{j}+z_{j}^{-1}) + \Phi_n(z)
        \end{equation}
        where the $c_{j,n}$'s  are nonzero real constants independent of $z$ and  $\Phi_n$ is a Laurent polynomial with real coefficients depending solely on the variables $\{z_i : p_i < p_j\}$. In particular, if $p_j=p_0$, then $\Phi_n$ is a real constant.
\end{enumerate}
    If $p_0=1$, then $\Phi_n=0$ and $c_{j,n} = 1$ for each $j$ such that $p_j = p_0$.
\end{lemma}

\begin{proof}
The statements in the case $p_0 \geq 2$ follow from Corollary~\ref{coro:loopExpansion} while the statements for $p_0=1$ follow from computations with the help of the Feynman--Hellmann theorem.
\end{proof}

This can then be put to use to get suitable estimates.

\begin{lemma} \label{lem:bounds}
    Assume $V:\bbZ^d \to \bbR$ is $p$-periodic and  non-degenerate, let $\rho_0>0$ be given, define $\varepsilon_0$ as in Lemma~\ref{lem:analyticity}, and let $p_0 := \min\{ p_1,\ldots,p_d\}$.
\begin{enumerate}[label={\rm(\alph*)}]
    \item \label{item:imaginaryPart}
     There is a constant $C_1 = C_1(d, p, \separate V, \rho_0) > 0$ such that
    \begin{equation}
        |\Im \eta_n(\varepsilon,z)| \leq C_1 \varepsilon^{ p_0}
    \end{equation}
    for all $|\varepsilon| \leq \varepsilon_0$ and $z \in \overline{\annulus(\rho_0)}$, and $n \in W$.
\medskip

\item \label{item:QepsBound}
There is a constant $C_2 = C_2(d, p, \separate V, \rho_0) > 0$ such that
\begin{equation}
\max\{\|I - Q_\varepsilon(z)\|,\|I-[Q_\varepsilon(z)]^{-1}\|\} \leq C_2\varepsilon
\end{equation}
for all $|\varepsilon| \leq \varepsilon_0$ and $z \in \overline{\annulus(\rho_0)}$.
\end{enumerate}
\end{lemma}

\begin{proof}
\ref{item:imaginaryPart}
On account of Lemma~\ref{lem:etaBounds}, we deduce that
\begin{equation}
    \eta_n(\varepsilon,z) 
    = V(n) +  \sum_{r=1}^\infty \eta_n^{(r)}(z)\varepsilon^r,
\end{equation}
where the $\eta_n^{(r)}$ are analytic and each is real and constant for $r < p_0$.
The desired bound follows presently.
The dependence of the constant follows from the form of $\eta_n^{(r)}$ in Theorem~\ref{t:rayleigh} and induction.
\bigskip

\ref{item:QepsBound}   Noting that $Q_\varepsilon(z)$ is  analytic as a function of $(\varepsilon,z)$ with $Q_0 = I$, this follows directly.
The dependence of the constant follows from the form of $u_n^{(r)}$ in Theorem~\ref{t:rayleigh} and induction.
\end{proof}

\subsection{Proofs of Main Results}

We now put everything together to prove the main results.

\begin{proof}[Proof of Theorem~\ref{t:asymptotic}]
From the definitions, we see that
\begin{equation} \label{eq:vAsyRelationGi}
    v_\asy(A_\varepsilon,\psi) = \|(G_1\psi,\ldots,G_d\psi)\|,
\end{equation}
where $G_i$ denotes the $i$th component of the group velocity, which in turn is given by:
    \begin{align} \nonumber
        G_i&:= \operatorname{s-lim}_{t \to \infty} \frac{1}{t} X_i(t)
         = \operatorname{s-lim}_{t \to \infty} \frac{1}{t} e^{\iop t A_\varepsilon}X_i e^{-\iop t A_\varepsilon} \\
        & \qquad = \frac{p_i}{|\bbT^d|}  \floquet^{-1} \left[ \int^\oplus_{\bbT^d} \sum_n \frac{\partial \eta_n(\varepsilon,e^{\iop \theta})}{\partial \theta_i} \Pi_n(\varepsilon, e^{\iop \theta}) \, \rmd \theta  \right] \floquet
        \label{eq:vasy}
    \end{align}
    in the strong sense, where we recall 
    $X_i$ denotes the position operator $X_i \psi(x) = x_i \psi(x)$, $\eta_n(\varepsilon,e^{\iop \theta})$ denotes the eigenvalue of $A_\varepsilon(e^{\iop \theta})$ that is analytically continued from $V(n)$, and $\Pi_n(\varepsilon, e^{\iop \theta})$ denotes projection onto the corresponding eigenspace; see e.g.\ \cite{Fillman2021OTAA} for a proof.
    In particular, $G_i$ is, up to a unitary conjugation, given by a generalized multiplication operator and therefore
    \begin{equation} \label{eq:GiNorm}
        \|G_i\|
        = p_i \max\left\{ \left| \frac{\partial \eta_n}{\partial \theta_i}(\varepsilon, e^{\iop \theta}) \right| : \theta \in \bbT^d, \ n \in W \right\}.
    \end{equation}

     For each $i$, we note from Lemma~\ref{lem:etaBounds} that one has
    \begin{equation} \label{eq:eta_nDerivExpansion}
        \frac{\partial \eta_n}{\partial \theta_i} = 
            2c_{i,n} \sin( \theta_i) \varepsilon^{p_i} + O(\varepsilon^{p_i + 1}) ,
    \end{equation}
    with $c_{i,n} \neq 0 $.
    Putting together \eqref{eq:GiNorm} and \eqref{eq:eta_nDerivExpansion}, 
    \begin{equation} \label{eq:GinormExpansion}
        \|G_i\| = \widetilde c_i \varepsilon^{p_i} + O(\varepsilon^{p_i+1}), \quad i = 1,2,\ldots, d,
    \end{equation}
    where $\widetilde c_i = 2p_i \max_n |c_{i,n}| \neq 0$.

Putting together \eqref{eq:GinormExpansion} and \eqref{eq:vAsyRelationGi} (and taking the maximum over $i$) gives $v_\asy(A_\varepsilon) = c\varepsilon^{p_0} + O(\varepsilon^{p_0+1})$.
The result for $v_\asy(H_\mu)$ in \eqref{eq:vAsyBounds} follows after rescaling $H_\mu = \mu A_{1/\mu}$.
The result for $v_\asy(H_\mu, \delta_0)$ follows from the computations above and the observations $\floquet \delta_0 = e_0$ and $\Pi_n = e_ne_n^* + O(\varepsilon)$, viz:
\begin{align*}
    \|G_i\delta_0\|^2
    & = 
     p_i^2   \left\| \left[ \int^\oplus_{\bbT^d} \sum_n \frac{\partial \eta_n(\varepsilon,e^{\iop \theta})}{\partial \theta_i} \Pi_n(\varepsilon, e^{\iop \theta})  \, \frac{\rmd \theta}{|\bbT^d|}  \right] e_0 \right\|^2 \\
     & =
     p_i^2     \int_{\bbT^d}  \sum_n\left\| \frac{\partial \eta_n(\varepsilon,e^{\iop \theta})}{\partial \theta_i} \Pi_n(\varepsilon, e^{\iop \theta}) e_0 \,  \right\|^2 \frac{\rmd \theta}{|\bbT^d|} \\
     & =
     p_i^2     \int_{\bbT^d}  \left| \frac{\partial \eta_0(\varepsilon,e^{\iop \theta})}{\partial \theta_i} \right|^2(1 + O(\varepsilon))\, \frac{\rmd \theta}{|\bbT^d|} \\
     & =
     p_i^2     \int_{\bbT}  \left| 2c_{i,0} \sin\theta_i \varepsilon^{p_i} + O(\varepsilon^{p_i+1})\right|^2(1 + O(\varepsilon))\,  \frac{\rmd \theta_i}{|\bbT|} \\
     & = 4p_i^2c_{i,0}^2(\tfrac{1}{2}) \varepsilon^{2p_i} + O(\varepsilon^{2p_i+1}),
\end{align*}
which shows $v_\asy(A_\varepsilon,\delta_0) = c \varepsilon^{p_0} + O(\varepsilon^{p_0+1})$ after taking the maximum over $i$.
Rescaling to $H_\mu = \mu A_{1/\mu}$ as before concludes the argument.
\end{proof}

\begin{proof}[Proof of Theorem~\ref{t:LiebRob}]
Fix $\rho_0>0$, define $\varepsilon_0$ as in Lemma~\ref{lem:analyticity}, and let us consider $A_\varepsilon = \varepsilon \Delta+V$ with $|\varepsilon| \leq \varepsilon_0$. 
In order to bound (block) matrix elements,
we use \eqref{eq:duhamel} and use the substitution $z_j = \exp(\iop \theta_j)$, $\rmd z_j/z_j = \iop \, \rmd \theta_j$ to transform to an iterated contour integral:
\begin{align}
\nonumber
e^{-{\iop}tA_\varepsilon}(x,y)
& =  \int_{\bbT^d} e^{-{\iop} t A_\varepsilon(e^{ {\iop} \theta})} e^{-{\iop}\langle (x-y), \theta \rangle} \, \frac{\rmd\theta}{|\bbT^d|} \\
\label{eq:unitaryGroupIteratedContour}
& = \frac{1}{\iop^d|\bbT^d|}
 \oint_{|z_1|=1} \cdots \oint_{|z_d|=1}
e^{-\iop t A_\varepsilon(z)} z^{-(x-y+\mathbf{1})}\, \rmd z,
\end{align}
where we used the multi-index notation $z^n = z_1^{n_1}\cdots z_d^{n_d}$ for $z \in (\bbC^*)^d$, $n \in \bbZ^d$, and $\mathbf{1}=(1,1,\ldots,1)$.
\bigskip

\noindent\textbf{Upper Bound.}
The first step is to deform the contours as follows:
for each $1 \le j \le d$, define
\begin{equation}
    \sigma_j 
    = \sgn(x_j-y_j) 
    = \begin{cases}
        -1 & x_j < y_j, \\ 
        0  & x_j = y_j, \\
        1  & x_j > y_j.
    \end{cases}
\end{equation}
We then shall pass to the contours $|z_j|=e^{\sigma_j \rho_0}$.
Since the integrand is holomorphic in $\annulus(2\rho_0)$, we have
\begin{align*}
& \frac{1}{\iop^d|\bbT^d|}
 \oint_{|z_1|=1} \cdots \oint_{|z_d|=1}
e^{-\iop t A_\varepsilon(z)} z^{-(x-y+\mathbf{1})}\, \rmd z \\
& \qquad  = \frac{1}{\iop^d|\bbT^d|} \oint_{|z_1| = e^{\sigma_1 \rho_0}} 
\cdots \oint_{|z_d|=e^{\sigma_d \rho_0}}
e^{-\iop t A_\varepsilon(z)} z^{-(x-y+\mathbf{1})}\, \rmd z.
\end{align*}
Estimating, diagonalizing, and using Lemma~\ref{lem:bounds}.\ref{item:QepsBound} to estimate $\|Q_\varepsilon^{\pm 1}(z)\|$, we deduce for $t \geq 0$:
\begin{align*}
    & \|e^{-{\iop}tA_\varepsilon}(x,y)\| \\
    & \qquad \leq  \frac{1}{|\bbT^d|}  \oint_{|z_1| = e^{\sigma_1 \rho_0}} 
    \cdots \oint_{|z_d| = e^{\sigma_d \rho_0}} \|Q(z)\|\| \diag(\exp(-{\iop}t \eta_n(\varepsilon, z)))\| \times \cdots \\ & \hspace{2in} \cdots \times\| [Q(z)]^{-1}\|  |z|^{-(x-y+\mathbf{1})}\, |\rmd z|, \\
    & \qquad\leq 
    (1+C_2\varepsilon+O(\varepsilon^2))^2  \exp\left(t \max_{n,z} |\Im \eta_n(\varepsilon,z)| \right)\prod_{j=1}^d e^{-\rho_0|x_j-y_j|}.
\end{align*}
Using Lemma~\ref{lem:bounds}.\ref{item:imaginaryPart}, we deduce
\begin{align}
   \|e^{-{\iop}tA_\varepsilon}(x,y)\|\leq  C   e^{C_1|t|\varepsilon^{p_0}} e^{-\rho_0|x-y|_1},
\end{align}
which concludes the proof of the upper bound after rescaling to $H_\mu = \mu A_{1/\mu}$ and using Remark~\ref{rem:blocksToEntries} to pass from blocks to matrix elements.
\bigskip

\noindent  \textbf{Lower Bound.} The lower bound follows readily from \eqref{eq:vAsydeltaBounds} and a straightforward generalization of \cite[Theorem~A.1]{ADFS} to higher dimensions (the formulation in \cite{ADFS} is for $d=1$, but the proof applies in higher dimensions with cosmetic changes).
\end{proof}

\bibliographystyle{abbrv}
\bibliography{refs}

\end{document}